\documentclass[copyright,creativecommons,noderivs]{eptcs}
\providecommand{\event}{GandALF 2011}

\title{Separation of Test-Free Propositional Dynamic Logics over Context-Free Languages}
\def\titlerunning{Separation of Test-Free Propositional Dynamic Logic over Context-Free Languages}
\author{Markus Latte\footnote{Supported by the \acronym{DFG} Graduiertenkolleg 1480 (\acronym{PUMA}).} \institute{Department of Computer Science \\University of Munich, Germany}}
\def\authorrunning{Markus Latte}

\usepackage{microtype}
\usepackage{relsize}
\usepackage{amsmath,amssymb,amsthm}
\usepackage{xspace}
\usepackage{array}
\usepackage{pgf}

\input{theorems.def}
\input{math.def}
\input{prelim.def}
\input{language.def}
\input{measure.def}
\input{emptyword-free.def}
\input{elimA.def}
\input{intro.def}

\begin{document}
\maketitle

\begin{abstract}
For a class $\mathfrak L$ of languages let \PDL[$\mathfrak L$] be an
extension of Propositional Dynamic Logic which allows programs to be
in a language of $\mathfrak L$ rather than just to be
regular.  If $\mathfrak L$ contains a non-regular language,
\PDL[$\mathfrak L$] can express non-regular properties, in contrast to
pure PDL.

  For regular, visibly pushdown and deterministic context-free
languages, the separation of the respective \PDL{}s can be proven by
automata-theoretic techniques.  However, these techniques introduce
non-determinism on the automata side.  As non-determinism is also the
difference between \DCFL and \CFL, these techniques seem to be
inappropriate to separate \PDL[\DCFL] from \PDL[\CFL].  Nevertheless, 
this separation is shown but for programs without test operators.
\end{abstract}

%
%
%

%
%
%
%
%
%
%
%

\section{Introduction} %
\label{sec:intro}

Propositional Dynamic Logic (\acronym{PDL})~\cite{FischerLadner79} is a logical formalism to 
specify and verify programs~\cite{DynamicLogic,LogicsOfPrograms90,HandbookPhil:Harel:DynamicLogic}.
These tasks rely on the satisfiability and model-checking problems.
Applications in the field are supported by their relatively low complexities: \EXPTIME- and \PTIME-complete, respectively~\cite{FischerLadner79}.

Formulas in \acronym{PDL} are interpreted over labeled transition systems.
For instance, the formula $\langle p \rangle \varphi$ means that
after executing the program $p$ the formula $\varphi$ shall hold.  In this context, programs
and formulas are defined mutually inductively.  This mixture allows programs to test whether or not a formula holds at
the current state.  Additionally, programs are required to be regular over the set of atomic programs and test operations.
For instance, the program {\tt while (b) do p;} can be rendered as $\langle (b? ; p)^\ast ; \neg b\rangle \varphi$
to ensure that the loop is finite and that $\varphi$ holds when the loop terminates~\cite{FischerLadner79}.

The small model property of \acronym{PDL}~\cite{FischerLadner79} cuts both ways.  
First, it admits a decision procedure for satisfiability, but secondly it
restricts the expressivity to regular properties.  As a consequence counting properties and, in particular,
the nature of execution stacks cannot be expressed.  The last consequence runs contrary to the verification of recursive programs.

A natural way to enhance the expressivity is to relax the regularity requirement.
For a class $\mathfrak L$ of languages let \PDL[$\mathfrak L$] denote the variation which requires that
any program belongs to $\mathfrak L$\footnote{
If test operations and deterministic languages are in involved, the test operations also
must behave deterministically.  In the case of \DCFL{}s the additional restriction reads as follows (using the notation in~\cite{HopcroftUllman79}).
\begin{list}{$\bullet$}{%
     \setlength{\itemsep}{-.5ex}
     \setlength{\parsep}{3pt}
     \setlength{\topsep}{1pt}
     \setlength{\partopsep}{0pt}
     \setlength{\leftmargin}{1.5em}
     \setlength{\labelwidth}{1em}
     \setlength{\labelsep}{0.5em} }
\item For any state $q$, at most one of $\delta(q, a, X)$ (for $a \in \Sigma$), $\delta(q, \emptyword, X)$
      and $\delta(q, \varphi?, X)$ (for some \PDL[$\mathfrak L$]-formula $\varphi$) is not empty.
\item For any state $q$ and two distinct \PDL[$\mathfrak L$]-formulas $\varphi_1$ and $\varphi_2$, we have that if
      $\delta(q, \varphi_1?, X) \neq \emptyset$ and $\delta(q, \varphi_2?, X) \neq \emptyset$
      then $\varphi_1$ and $\varphi_2$ are semantically disjoint, that is $\models \neg (\varphi_1 \wedge \varphi_2)$.
\end{list}
Otherwise, it would be possible to simulate a non-deterministic choice by inserting a test for ``true'' for every possible choice
and vary each test syntactically in a different way.  Note that ``true'' has infinitely many synonyms.
A non-example is \DeltaTestPDL[\CFL] $=$ \DeltaTestPDL[\DCFL] in \cite{nonregctl_lpar_2010}.%
}.
For instance, we write a diamond as $\langle L \rangle \varphi$
for $L \in \mathfrak L$.
This leads to a hierarchy of logics.  Obviously, \PDL[$\mathfrak L$] $\leq$ \PDL[$\mathfrak M$] holds
for $\mathfrak L \subseteq \mathfrak M$.
Besides regular languages, we consider the variations for the class of visibly pushdown languages~\cite{AM04}, \VPL,
the class of deterministic context-free languages, \DCFL, and context-free languages, \CFL.
The inclusion order continues on the logics' side.
\begin{align}\label{eq:PDL chain}
  \text{\PDL}
  =
  \text{\PDL[\REG]}
  \leq
  \text{\PDL[\VPL]}
  \leq
  \text{\PDL[\DCFL]}
  \leq
  \text{\PDL[\CFL]}
  \text.
\end{align}
Harel et al.\ discussed the effect of adding \emph{single} (deterministic) context-free programs
to \acronym{PDL}~\cite{HarelPnueliStavi83,HarelRaz90,DynamicLogic}.  The logic \PDL[\VPL]
were introduced by L\"oding et al.~\cite{LLS07}.

To handle the respective decision problems, the languages are represented by a machine model for the respective class.
For each of these logics, any of its formula $\varphi$
can be translated into an $\omega$-tree-automaton which recognizes exactly
all tree-like models of $\varphi$ where the out-degree of any node is globally bounded.
Such a model exists iff $\varphi$ is satisfiable. %
For \PDL and \PDL[\REG] these tree-automata are finite-state~\cite{VardiWolper86},
for \PDL[\VPL] they are visibly pushdown tree-automata~\cite{HarelRaz90,LLS07} and
for \PDL[\DCFL] and \PDL[\CFL] they are tree-automata with unbounded number of stacks.
The last notion is rather artificial.  However, the stacks are used,
first, to accumulate unfulfilled eventualities and to simulate the complementation 
of programs given as pushdown automata.
Note that in the setting of visibly pushdown automata, only one stack suffices as 
$\omega$-\VPL{}s are closed under complementation~\cite{AM04}
and under determinisation (for stair-parity conditions)~\cite{conf/fsttcs/LodingMS04}.

The first two inequalities in~\eqref{eq:PDL chain} are strict.
In this paragraph we sketch the proofs for the first two inequalities.
Consider the language $L := \{c^n r^n \mid n \in N\}$ over
an alphabet $\Sigma \supseteq \{c, r\}$.  Hence, we have $L \in \text{\VPL}$ 
if we take $c$ for a call and $r$ for a return in a visibly pushdown alphabet for $\Sigma$.
Now, we claim that $\vartheta {:=} \langle L \rangle p$ is not expressible in \PDL[\REG] where $p$ is a proposition.
For the sake of contradiction, assume that there were such a formula.  Restricted to linear models, 
the previous translation leads to a finite-state \Buchi-automaton~$\mathcal A$ which recognizes those models.
Let $N$ be sufficiently large---which depends on the pumping length and the, here omitted, encoding.
Consider the following model of $\vartheta$ for $N'=N$.
\\
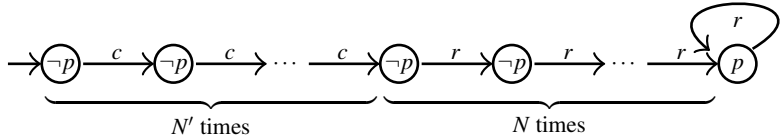
\hspace*{2em}
\begin{pgfpicture}{-1.0cm}{-1cm}{10cm}{1.0cm}
  \pgfsetlinewidth1pt
  \pgfsetendarrow{\pgfarrowto}

  \pgfxyline(-0.7,0)(-0.3,0)
  \mytext(0.0,0)[center,center]{$\neg p$}  \pgfcircle[stroke]{\pgfxy(0,0)}{2.5mm}
  \pgfxyline(0.3,0)(1.2,0) \mytext(0.75,0)[center,bottom]{$c$}
  \mytext(1.5,0)[center,center]{$\neg p$}  \pgfcircle[stroke]{\pgfxy(1.5,0)}{2.5mm}
  \pgfxyline(1.8,0)(2.7,0) \mytext(2.25,0)[center,bottom]{$c$}
  \mytext(3.0,0)[center,center]{\everymath{}$\cdots$}  
  \pgfxyline(3.3,0)(4.2,0) \mytext(3.75,0)[center,bottom]{$c$}
  \mytext(4.5,0)[center,center]{$\neg p$}  \pgfcircle[stroke]{\pgfxy(4.5,0)}{2.5mm}
  \pgfxyline(4.8,0)(5.7,0) \mytext(5.25,0)[center,bottom]{$r$}
  \mytext(6.0,0)[center,center]{$\neg p$}  \pgfcircle[stroke]{\pgfxy(6.0,0)}{2.5mm}
  \pgfxyline(6.3,0)(7.2,0) \mytext(6.75,0)[center,bottom]{$r$}
  \mytext(7.5,0)[center,center]{\everymath{}$\cdots$}  
  \pgfxyline(7.8,0)(8.7,0) \mytext(8.25,0)[center,bottom]{$r$}
  \mytext(9.0,0)[center,center]{$p$}  \pgfcircle[stroke]{\pgfxy(9.0,0)}{2.5mm}

  \pgfmoveto{\pgfxy(9.2,0.2)}
  \pgfcurveto{\pgfxy(10.5,1)}{\pgfxy(7.6,1)}{\pgfxy(8.6,0.2)}
  \pgfstroke
  \mytext(9.0,0.7)[center,top]{$r$}

  \mytext(2.0,-0.2)[center,top]{\everymath{}$\underbrace{\hspace*{5.5cm}}$} 
  \mytext(2.0,-0.7)[center,top]{$N'$ times}

  \mytext(6.5,-0.2)[center,top]{\everymath{}$\underbrace{\hspace*{5.5cm}}$} 
  \mytext(6.5,-0.7)[center,top]{$N$ times}

\end{pgfpicture}
\\
As $\mathcal A$ accepts this model, it also accepts this transition system for $N'<N$
due to the pumping lemma.  However, this structure is not a model of $\vartheta$.
The separation of \PDL[\VPL] and \PDL[\DCFL] can be achieved in similar fashion.
Take as program $L {:=} \{w \sharp w^R \mid w \in (\Sigma \setminus \{\sharp\})^\ast \} \in \text{\DCFL}$
over an alphabet $\Sigma \ni \sharp$.  For any visibly pushdown alphabet for
$\Sigma$ its return-part is not empty in general.  Using such a letter for the
$w$-part in $L$, an assumed visibly pushdown automaton for $\langle L \rangle p$
operates on that part like a finite-state automaton.  The same argumentation applies as for the first separation.

The separation for the last inequality in~\eqref{eq:PDL chain} is more cumbersome and intrinsic:
For the satisfiability problem, the emptiness problem for finite-state and 
for visibly pushdown tree-automata is decidable~\cite{Rabin:70,VardiWolper86,Thomas:InfiniteObjects90}\cite{conf/fsttcs/LodingMS04}.
The emptiness problem for the tree-automata with an unbounded number of stacks
can be considered as the halting problem for \Buchi-Turing machines~\cite{StaigerHandbook}.
Indeed, the satisfiability problems for \PDL[\DCFL] and \PDL[\CFL]
are $\Sigma^1_1$-complete~\cite{HarelPnueliStavi83}.
Hence, both logics are not distinct by a ``trivial'' reason.

The standard translation~\cite{VardiWolper86,HarelRaz90,LLS07} from formulas to tree-automata bases on Hintikka-sets.
For a fixed formula $\vartheta$ and for every node of the given transition system 
the automaton for $\vartheta$ guesses---among other things---the set of those subformulas of $\vartheta$ which hold at that node.
Informally speaking, the non-determinism is required to handle disjunctions in the given formula
and to recognize the termination of a program in an expression such as $\langle L \rangle \varphi$.
Note that a language in \DCFL might be not prefix-free.
However, non-determinism is also the difference between \DCFL and \CFL.
Hence, the translation seems not to suffice to separate \PDL[\DCFL] from \PDL[\CFL].

In this paper we make a step towards the separation of \PDL[\CFL] from \PDL[\DCFL].
For technical reasons we consider \PDL[$\mathfrak L$] without the test operations like $\varphi?$---call the logic \PDLplain[$\mathfrak L$]---and prove
the separation of the corresponding logics.  
This restriction is proper as \PDLplain is weaker than \PDL~\cite{Berman81}.
Note that \PDLplain[$\mathfrak L$] is exactly the $\EF$/$\AG$-fragment of \XCTL{$\mathfrak L$}~\cite{nonregctl_lpar_2010,nonregctl_arxiv_2010}.
This logic is obtained from \CTL by restricting the moments of until- and release-operations by
languages in $\mathfrak L$.
The separation of \XCTL{\DCFL} and \XCTL{\CFL} is unknown as well.

%

\section{Preliminaries}
\label{sec:prelim}

Let $\Sigma$ be an alphabet.  For a finite word $w \in \Sigma^\ast$ we
write $|w|$ for its length and $w\substr{i}{j}$ for its subword starting at
index $i$ and ending at index $j$ where $0\leq i \leq j < |w|$. Both
indices are zero-based.  For words $u, v\in \Sigma^\ast \cup \Sigma^\omega$ 
their concatenation is written as $uv$ and the reversal of $u$ as $u^R$.  
Concatenation is extended to sets in the usual way.
The empty word is denoted by $\emptyword$.  
A word $u \in \Sigma^\ast \cup \Sigma^\omega$ is a (proper) prefix of $w \in \Sigma^\ast \cup \Sigma^\omega$ 
iff there is $v \in \Sigma^\ast \cup \Sigma^\omega$ such that $uv=w$ (and $v\neq\emptyword$).
The notation of a suffix is defined similarly. 
For two languages $L_1$ and $L_2$ their left quotient
$L_1 \setminus L_2$ is $\{v \mid \exists u \in L_1 . uv \in L_2 \}$.
If one of both languages is a singleton we may
replace the language by its single word.
Standard notations are used~\cite{HopcroftUllman79} for
(deterministic) pushdown automata on finite words, \DPDA and \PDA, and (deterministic)
context-free languages.  Deterministic pushdown automata on
$\omega$-words, $\omega$\DPDA, are equipped with \Buchi-acceptance
conditions~\cite{StaigerHandbook}.

Let $\Prop = \{p,q,\ldots\}$ be a set of \emph{propositions}. %
A \emph{labeled transition system}, \LTS, is a triple $\mathcal T = (\States, \Transition{}, \ell)$
consisting of a set of states $\States$, of a labeled edge relation $\Transition{} \subseteq \States \times \Sigma
\times \States$ and of an evaluation function $\ell:\States \to 2^\Prop$.  We write
$s \Transition{a} t$ instead of $(s,a,t) \in \Transition{}$.  
A \emph{path} is a sequence $s_0,a_1,s_1,a_1,\ldots a_{n-1},s_n$ for some $n \in \Nat$ such that
$s_i \Transition{a_{i+1}} s_{i+1}$ for all $i \in \{0, \ldots, n-1\}$.
For such a path we may write $s_0 \Transition{a_0} s_1 \Transition{a_1} \ldots \Transition{a_{n-1}} s_n$.
A structure is a pair $\mathcal M = (\mathcal T, s)$ of an \LTS and a state in it, called \emph{root}.
Previous notations for \LTS{s} are also used for structures.
A structure $\mathcal M = ((\States', \Transition{}', \ell'), s')$ is an \emph{extension} of $\mathcal M$, written as $\mathcal M \leq \mathcal M'$,
iff $\States \subseteq \States'$, $\Transition{} \subseteq \Transition{}'$, $\ell$ is the restriction of $\ell'$ to $\States$, and $s=s'$.

Let $\mathfrak L$ be a class of languages.
We define the logic \PDLplain[$\mathfrak L$] in negation normal form using a \acronym{CTL}-like syntax~\cite{nonregctl_lpar_2010}---%
that is, $\EF^L \varphi$ stands for the \PDL-expression $\langle L \rangle \varphi$ for instance.
The formulas are given by the grammar
\begin{displaymath}
  \varphi \enspace ::= \enspace \enspace \False \enspace \mid \enspace \True \enspace \mid 
    \enspace p \enspace \mid \enspace \neg p \enspace \mid 
    \enspace \varphi \vee \varphi \enspace \mid \enspace \varphi \wedge \varphi \enspace \mid 
    \enspace \EF^{L} \varphi \enspace \mid \enspace \AG^{L} \varphi
\end{displaymath}
where $p \in \Prop$ and $L \in \mathfrak L$.
Such formulas are denoted by $\varphi$, $\psi$, $\vartheta$, and $\delta$.
The atoms $\False$ and $\True$ are called \emph{constants}, and $p$ and $\neg$ are called \emph{literals}.
Implication and equivalence are definable.
A formula $\EF^L \varphi$ is called \emph{$\EF$-formula}.  An \emph{$\AG$-formula} is meant analogously.
A formula is interpreted over a structure as follows. %
\\[1ex]
\begin{tabular}[t]{@{\qquad}l}
$\mathcal T, s \not\models \False$ \qquad\qquad%
$\mathcal T, s \models \True$ \qquad\qquad%
$\mathcal T, s \models p$ iff $p \in \ell(s)$ \qquad\qquad%
$\mathcal T, s \models \neg p$ iff $p \not\in \ell(s)$\\[1ex]
$\mathcal T, s \models \varphi_1 \vee \varphi_2$ iff $\mathcal T, s \models \varphi_1$ or $\mathcal T, s \models\varphi_2$\qquad\qquad%
$\mathcal T, s \models \varphi_1 \wedge \varphi_2$ iff $\mathcal T, s \models \varphi_1$ and $\mathcal T, s \models\varphi_2$\\[1ex]
$\mathcal T, s \models \EF^L \varphi$ iff 
there is path $s_0 \Transition{a_0} s_1 \Transition{a_1} \ldots \Transition{a_{n-1}} s_n$ with $s=s_0$,\;
$a_0 \cdots a_{n-1} {\in} L$ and $\mathcal T, s_n \models \varphi$
\\[1ex]
$\mathcal T, s \models \AG^L \varphi$ iff 
for all paths $s_0 \Transition{a_0} s_1 \Transition{a_1} \ldots \Transition{a_{n-1}} s_n$ with $s=s_0$
and $a_0 \cdots a_{n-1} \in L$: $\mathcal T, s_n \models \varphi$
\end{tabular}
\\[1ex]
If $\mathcal T, s \models \varphi$ then the structure $(\mathcal T, s)$ is a \emph{model} of $\varphi$.
A structure $(\mathcal T, s)$ is \emph{tree-like} iff $\mathcal T$ forms a tree with root $s$.
Since \PDLplain[$\mathfrak L$] is closed under bisimulation, every satisfiable formula has 
a tree-like structure as a model.  A formula $\varphi$ is a \emph{tautology}, written as $\models\varphi$,
iff every structure is a model of $\varphi$.

\section{Outline of the Proof}
\label{sec:goal}

For the following parts, fix an alphabet $\Sigma$ which at least contains $0$ and $1$ 
but not $\$$, and set $\Sigma_\$ := \Sigma \cup \{\$\}$.
The language of palindromes is denoted by $\Palindromes := \{ w \in \Sigma^\ast  \mid w = w^R \}$.
We will show that there is no \PDLplain[\DCFL]-formula
which is equivalent to the \emph{reference} \PDLplain[\CFL]-formula
$\EF^{\Palindromes \; \$} \True$.
As the reference formula does not contain propositions we
may assume that neither does any equivalent formula.
Equivalently, we may assume that $\Prop = \emptyset$.

For the sake of contradiction, 
let $\vartheta \in \text{\PDLplain[\DCFL]}$ be a \emph{candidate} formula which is assumed 
to be equivalent to $\EF^{\Palindromes \; \$} \True$. 
To illustrate the main problem about provoking a contradiction, 
we begin with a simpler setting in which $\vartheta$ does not contain any conjunctions or $\AG$-formulas.
As we have the equivalences
\begin{align*}
  \EF^{L} \False &\leftrightarrow \False\text,
&
  \EF^{L_1}\EF^{L_2}\psi &\leftrightarrow \EF^{L_1L_2}\psi\text{, and}
&
  \bigvee_i \EF^{L_i} \psi &\leftrightarrow \EF^{\bigcup_i L_i} \psi
\end{align*}
the formula $\vartheta$ can be rewritten as 
\begin{align*}
  \EF^{\bigcup_i L_{i,1} \cdots L_{i,n_i}} \True
\end{align*}
where $L_{i,j}$ are \DCFL{}s over $\Sigma_\$$.  
In general, an equivalence $\EF^{\Palindromes \; \$} \True \leftrightarrow \EF^L \True$ implies
\[
  \Palindromes = \{w \in \Sigma^\ast \mid w\$ \text{ is a prefix of a word in }L\}
\]
for $L\subseteq \Sigma_\$^\ast$.
Therefore, we have that $\Palindromes$ would be expressible as
a finite union over a finite concatenation over \DCFL{}s over $\Sigma$.
Some combinatorial argument shows that this is impossible.

Back to the real world, we are also faced with conjunctions and $\AG$-formulas in $\vartheta$.
A natural attempt is to eliminate these subformulas.  Indeed, a conjunction
seems not to support a statement which speaks about a single path only.
Instead, it speaks about a bunch of paths.  Similarly, an $\AG$-formula
is not monotone with respect to models but the reference formula is monotone.
To turn off such formulas, one could saturate the considered structures
with substructures which falsify $\AG$-formulas and
which do not affect the desired property $\EF^{\Palindromes \; \$}\True$.
However on such a new structure, the attached substructures could be
recognized by other $\EF$-subformulas.  But these subformulas 
need not to be concerned with palindromes in any reasonable way.
Moreover, Boja\'{n}czyk proved~\cite{Bojanczyk:commonFragment}%
---for the dual setting---that such an elimination procedure is only possible
if---in our setting---palindromes were expressible as a finite union
of languages of the form $A_0^\ast a_1 A_1^\ast a_2 \cdots A_{n-1}^\ast a_n A_n^\ast$  %
for $a_1, \ldots, a_n \in \Sigma$ and $A_0, \ldots, A_n \subseteq \Sigma$.
Obviously, this is not the case.

Therefore, our strategy is different. %
First, we show that topmost $\AG$-formulas and topmost conjunctions can be eliminated 
(\S~\ref{sec:emptyword-free} and~\ref{sec:elim}).
This renders the candidate formula $\vartheta$ as $\bigvee_i \EF^{L_i} \psi_i$ for some $L_i \subseteq \Sigma_\$^+$
and some formulas $\psi_i$ with unknown structure.
Secondly, if $L_i$ is not a singleton language then the formula $\EF^{L_i} \psi_i$ per se provides all the information required for a contradiction.  Either it under- or over-approximates palindromes.
And if $L_i$ is a singleton we proceed in a similar way with the left-quotient of $\vartheta$ with
the only word in $L_i$.
The whole procedure (\S~\ref{sec:extraction}) terminates through a sophisticated measure (\S~\ref{sec:measure}).
The case that $L_i$ is not a singleton give rise to a characterization of languages
which will bridge between the formula and the language part of the separation proof.

\begin{definition}\label{def:good}
  A language $L \subseteq \Sigma^\ast$ is \emph{good} iff
  $%
    L = \bigcup_{i \in I} L_i R_i
  $ %
  such that $I$ is finite, and for each $i \in I$, the language $L_i$ is a \DCFL, $|L_i| \geq 2$ and $R_i \subseteq \Sigma^\ast$.
\end{definition}

In the view of Boja\'{n}czyk's result, our iterated elimination is non-uniform compared to the preferable approach in the previous paragraph. 
Finally, we show on the language-theoretical level that palindromes
are not good (\S~\ref{sec:palindromes_notin_dcfl}).

%
%
%

\section{On Palindromes and \texorpdfstring{\DCFL{}s}{DCFLs}}
\label{sec:palindromes_notin_dcfl}

In this section it is proven that the language of palindromes is not good.
For this purpose we first show that this language is not expressible as
a union of \DCFL{}s (Theorem~\ref{thm:dcfl neq palindromes}).  
Although it is know that the set of palindromes is not deterministic context-free,
the standard proof~\cite[Cor.~1]{GinsburgGreibach:dcfl} does not seem to be adaptable 
because the applied $\min$-operator does not commutate with the union.
As a second step, it is shown that 
if palindromes are underapproximated by a concatenation
then the components of the concatenation follow a very simple pattern (Lemma~\ref{lem:boundR}).


\begin{lemma}[Pumping lemma]\label{lem:pumpinglemma}
  Let $u \in \Sigma^\omega$ be accepted by an \omegaDPDA $\mathcal A$. %
  There are words $u_0 \in \Sigma^\ast$, $u_1 \in \Sigma^+$ and $u_2 \in \Sigma^\omega$ such that
  $u_0 u_1 u_2 = u$, and $u_0 u_2$ is accepted by $\mathcal A$.
\end{lemma}
\begin{proof}
  Firstly, we may assume that $\mathcal A$ only erases or pushes symbols from or on the stack
  and never changes the topmost symbol.  Indeed, an \omegaDPDA can keep the topmost element
  of the stack in its control state~\cite[Sect.~10.1]{HopcroftUllman79}. %
  By this restriction, in any run the stacks of two consecutive configuration are comparable with respect to the prefix-order.
  Secondly, consider the infinitely many stair positions in the accepting run of $\mathcal A$ on $u$.
  By a stair position~\cite{conf/fsttcs/LodingMS04} we understand a position such that the current stack content 
  is a prefix of all further stack contents in this run.  
  As the set of states is finite, %
  there are two different stair positions which name
  the same state.  We may assume that a non-empty part of $u$, say $u_1$ with $u=u_0 u_1 u_2$, 
  fits into their gap.
  Hence, this part can be removed.
  By the definition of stairs, the obtained sequence of configurations
  is a run of $\mathcal A$ on $u_0 u_2$.  
  As the modification affects a prefix of $u$ only, $\mathcal A$ also accepts $u_0 u_2$.
\end{proof}

\begin{theorem}\label{thm:dcfl neq palindromes}
  Let $v \in \Sigma^\ast$, $n \in \Nat$, and $L_1, \ldots, L_n$ be \DCFL{}s over $\Sigma$.
  Then $\bigcup_{i=1}^n L_i \neq v \backslash \Palindromes$.
\end{theorem}
\begin{proof}
  Define the sequence $(w_i)_{i \in \Nat}$ of strictly prefix-ordered words as follows.
  \begin{align*}
    w_0 &{:=} v^R
  \\
    w_{i+1} &{:=} w_i 1 0^{i} 1 w_i^R v^R & \text{($i \in \Nat$)}
  \end{align*}
  For all $i \in \Nat$ we have $w_i \in v \backslash \Palindromes$.
  For the sake of contradiction, assume that  
  \begin{align}\label{eq:dcfl neq palindromes:assumption}
    \bigcup_{i=1}^n L_i = v \backslash \Palindromes
    \text.
  \end{align}
  We sample the candidate on the left of Eq.~\ref{eq:dcfl neq palindromes:assumption}
  with the words $\{w_i\}_{i \in \Nat}$.  Since
  the union is finite, there is an infinite $I \subseteq \Nat$ and an $i \in \{1,\ldots,n\}$ such that
  the words $\{w_i\}_{i \in I}$ belong to $L_i$.
  Let $\mathcal A$ be a \DPDA for $L_i$.
  Additionally, we consider $\mathcal A$ as an \omegaDPDA where the final states are the \Buchi-states.
  Hence, as $\mathcal A$ is a deterministic device it accepts 
  \begin{align}
    w := \lim_{i \in \omega} w_i  = \lim_{i \in I} w_i \in  \Sigma^\omega
    \text.
  \end{align}
  Apply Lemma~\ref{lem:pumpinglemma} to $\mathcal A$ and $w$.
  Let $u_0$, $u_1$, $u_2$ be the obtained factors.  We
  run $\mathcal A$ on $w$ for at least $|u_0 u_1|$ steps until it
  processes some subword $1 0^\kappa 1$ for the first time.
  Note that the function which maps $i \in \Nat$ to the first occurrence of $1 0^i 1$ in $w$ is unbounded.
  Let $\ell$ be the first index in $w$ after that subword.  So
  far, $\mathcal A$ has seen the first $\ell$ letters in $w$.  We keep
  $\mathcal A$ running for at least another $\ell+|v|$ steps until it
  reaches a final state.  Such a run is always possible as $\mathcal
  A$ accepts infinitely many prefixes of $w$.
  Let $u'$ be the word constructed in this way.
  Hence, $u' \in v \backslash \Palindromes$ as $\mathcal A$ accepts $u'$.

  Let $u''$ be the word $u'$ where the $u_1$-block is removed.  That is
  $%
    u'' := u'\substr{0}{|u_0|-1} \; u'\substr{|u_0 u_1|}{|u'|-1}
  $. %
  Again by construction
  and Lemma~\ref{lem:pumpinglemma}, $\mathcal A$ accepts $u''$.
  Thus, $u'' \in v \backslash \Palindromes$.
  Let $\hat u$ be the word between $u_1$ and the block $1 0^\kappa 1$,
  that is 
  $%
    \hat u = w\substr{|u_0 u_1|}{\ell-3-\kappa}
  $. %
  As $v u'$ is a palindrome, it ends in the word $(v u_0 u_1 \hat u 1 0^\kappa 1)^R$ of length $\ell+|v|$.  
  The modification leading to $u''$ affects at most the first $\ell$
  positions only.  Hence, as $|u'| \geq 2 \ell + |v|$, $u''$ also
  ends in $(v u_0 u_1 \hat u 1 0^\kappa 1)^R$. As $v u''$ is also a
  palindrome, $u_0 \hat u 1 0^\kappa 1$ is a prefix of 
  $u_0 u_1 \hat u 1 0^\kappa 1$.  Since $u_1$ is not the empty word, this is a
  contradiction to the choice of $1 0^\kappa 1$.  
\end{proof}

%

\begin{lemma}\label{lem:orderL}
  If $L R \subseteq \Palindromes$ and $R$ is infinite then $L$ is prefix-ordered. %
\end{lemma}
\begin{proof}
  Let $\ell_0, \ell_1 \in L$ with $|\ell_0| \leq |\ell_1|$.
  Take $r \in R$ such that $|r| \geq |\ell_1|$.
  This is possible as $R$ is infinite.
  Since $\ell_0 r$ and $\ell_1 r$ are palindromes,
  $\ell_0^R$ and $\ell_1^R$ are suffixes of $r$.
  Therefore, $\ell_0$ is a prefix of $\ell_1$.  
\end{proof}

\begin{lemma}\label{lem:boundR}
  Suppose $L R \subseteq \Palindromes$, $|L| \geq 2$ and $R$ is infinite.
  Then 
  \[
    R \subseteq \hat u^* \hat U
  \]
  for some word $\hat u \in \Sigma^\ast$ and a finite language $\hat U \subset \Sigma^\ast$.
\end{lemma}

\begin{proof}
  Let $u_0, u_1$ be two distinct words in $L$.  By the Lemma~\ref{lem:orderL} we may assume that $u_0$ is a proper prefix of $u_1$.
  Define $\hat u := u_0 \backslash u_1$.  Note that $u_0 \hat u = u_1$.

  \begin{claim}
  For $w \in R$ and $n \in \Nat$ we have 
  \begin{enumerate}
     \item \label{lem:boundR:claim:prefix} $\hat u^n$ is a prefix of $w$, and
     \item $(\hat u^R)^n u_0^R$ is a suffix of $w$
  \end{enumerate}
  if $n |\hat u| + |u_0| \leq |w|$.
  \end{claim}
  \begin{proofofclaim}
  By induction on $n$ for a fixed $w \in R$.
  If $n=0$, $u_0^R$ is a suffix of $w$ as $u_0 w$ is a palindrome.
  For the step case from $n$ to $n+1$ assume that 
  \begin{align}
    (n+1) |\hat u| + |u_0| \leq |w|
    \text.
    \label{eq:boundR:1}
  \end{align}
  The word $v := u_0 \hat u^{n+1} = u_1 \hat u^n$ is prefix of $u_1 w$ by IH\ref{lem:boundR:claim:prefix}.
  As $u_1 w$ is palindrome, $v^R$ is a suffix of $w$ because of~\eqref{eq:boundR:1}.
  This proves the second item.
  Since $u_0 w$ is also a palindrome, it has $v$ is a prefix.
  Hence $\hat u^{n+1}$ is a prefix of $w$---this is the first item.
  \end{proofofclaim}
  Let $w \in R$.
  For $N_w := \floor{(|w| - |u_0|)/|\hat u|}$, the claim yields
  $w = \hat u^{N_w} \hat w$ where $\hat w$ are the $r_w {:=} |w| - N_w|\hat u|$ last letters of $w$.
  Since $r_w \leq |u_0| + |\hat u|$ is bounded independently of $w$,
  there is a finite set $\hat U$ such that $R \subseteq \hat u^* \hat U$.
\end{proof}

\begin{lemma}\label{lem:orthogonal word}
  Let 
  \begin{align}
    L := \bigcup_{i \in I} u_{i,0} \; u_{i,1}^\ast \; u_{i,2}^\ast \; u_{i,3}
  \end{align}
  for $I$ finite and $u_{i,j} \in \Sigma^\ast$ for all suitable indices.
  Then there is a word $w \in \Sigma^\ast$ which is not a prefix of any word in $L$.
\end{lemma}
\begin{proof}
  Consider the tree $\Sigma^\omega$. %
  For each $i \in I$, the word $u_{i,0} \; u_{i,1}^\omega $ defines a (finite or infinite) path in the tree.
  As $I$ is finite, there is a $w_0 \in \Sigma^\ast$ which is not on these paths. 
  There are at most $|w_0| \cdot |I|$ paths of the form 
  $u_{i,0} \; u_{i,1}^j \; u_{i,2}^\omega$ for $i \in I$ and $j \in \Nat$
  which pass $w_0$.  By the same argument, we get a word $w_1$ which
  extends $w_0$ and cannot be reached by these paths.  
  A final application to $w_1$ and $u_{i,0} \;  u_{i,1}^j \; u_{i,2}^k \; u_{i,3}$ 
  for $i \in I$ and $j,k \in \Nat$ yields the claimed word $w$.
\end{proof}

\begin{corollary}\label{cor:Palindromes not good}
  The set $\Palindromes$ is not good.
\end{corollary}
\begin{proof}
  For the sake of contradiction, assume the contrary, that is
  \begin{align}\label{eq:Palindromes not good:assumption}
    \Palindromes = \bigcup_{i \in I} L_i R_i
  \end{align}
  where $I$ is finite, and for any $i \in I$ the language $L_i$ is a \DCFL, and $|L_i| \geq 2$.  
  Set $I^+ {:=} \{ i \in I \mid R_i \text{ is finite} \}$, and $I^- {:=} I \setminus I^+$.
  For any $i \in I^+$, we have that $L_i R_i$ is a \DCFL~\cite[Thm.~3.3]{GinsburgGreibach:dcfl}
  as $R_i$ is finite in particular.
  \par
  Let $i \in I^-$.  Since $|L_i| \geq 2$ and $R_i$ is infinite, Lemma~\ref{lem:boundR}
  shows that
  $%
    R_i \subseteq r_i^\ast \widehat R_i 
  $ %
  for some $r_i \in \Sigma^\ast$ and for a finite language $\widehat R_i \subset \Sigma^\ast$.
  Depending on the size of $L_i$ we can bound $L_i R_i$.
  If $L_i$ is infinite then the very same lemma shows by reversal that
  $%
    L_i \subseteq \widehat L_i \ell_i^\ast
  $ %
  for some $\ell_i \in \Sigma^\ast$ and a finite $\widehat L_i \subset \Sigma^\ast$.
  Hence,
  \begin{align*}
    L_i R_i 
  \subseteq
    \widehat L_i \; \ell_i^\ast \; r_i^\ast \; \widehat R_i 
  = 
    \bigcup_{\substack{x \in \widehat L_i, y \in \widehat R_i}}
      x \; \ell_i^\ast \; r_i^\ast \; y
    \text.
  \end{align*}
  In the other case---$|L_i|$ is finite---one obtains
  \begin{align*}
    L_i R_i \subseteq &
    \bigcup_{\substack{x \in L_i, y \in \widehat R_i}} x \; r_i^\ast \; y
    \text.
  \end{align*}
  In both cases, the unions are finite.
  All in all, we have
  \begin{align*}
     \bigcup_{i \in I} L_i R_i 
     \quad
     &=
     \quad
     \bigcup_{i \in I^+} \underbrace{\;L_i \; R_i \;}_{\text{\DCFL}}
     \quad
     \cup
     \quad
     Q'
  \end{align*}
  where $
     Q' \subseteq Q \; := \; \bigcup_{i \in J} u_{i,0} \; u_{i,1}^\ast \; u_{i,2}^\ast \; u_{i,3}
  $
  for some finite set $J$, and some words $u_{i,0}, \ldots, u_{i,3} \in \Sigma^\ast$.
  By Lemma~\ref{lem:orthogonal word}, there is a finite word %
  $w$ which is not a prefix of any word in $Q$.
  Using~\eqref{eq:Palindromes not good:assumption}, we get
  \begin{align*}
    w \backslash \Palindromes 
    = 
    \bigcup_{i \in I^+} w \backslash (L_i R_i)
    \text.
  \end{align*}
  The left quotient with a single word $w$ is the inverse of the gsm mapping
  which sends a word $u$ to $w\, u$.  As \DCFL{}s are closed under the inverse of gsm 
  mappings~\cite[Thm.~3.2]{GinsburgGreibach:dcfl}, the language $w \backslash (L_i R_i)$
  is a \DCFL{} for $i\in I^+$.
  But this a contradiction to Theorem~\ref{thm:dcfl neq palindromes}.
\end{proof}

\section{A Measure for the Extraction}
\label{sec:measure}

Informally, the measure of a formula is a set of vectors.
Each vector measures the languages annotated to $\EF$-subformulas
along a path from the root of the formula to its atoms.
For the measure of a language, the size of its only word is considered
if the language is a singleton.

\begin{definition}\label{def:measure}
  Let $\Measure$ be the set of all finite subsets of $(\omega+1)^\ast$
  where $\omega+1 = \{0,1,2,\ldots,\omega\}$.
  The second argument of the cons-operator $\cdot :: \cdot$ on $(\omega+1)^\ast$ is
  extended to sets. %
  The empty list is written as $nil$. %
  The \emph{measure} of a formula is defined by
  \begin{align*}
    \measure(\ell) &{:=} \{ nil \}
    & \text{ for $\ell$ a literal or a constant}
  \\
    \measure(\varphi_0 \circ \varphi_1) &{:=} \measure(\varphi_0) \cup \measure(\varphi_1) 
    &\text{ for }\circ\in\{\wedge,\vee\}
  \\  
    \measure(Q^L \varphi) &{:=} ||L|| :: \measure(\varphi)
    & \text{ for }Q\in\{\EF,\AG\}
  \end{align*}
  where
  \[
    ||L|| = \begin{cases}
      |w| & \text{if $L=\{w\}$ for some $w \in \Sigma^\ast$,}
    \\
      \omega & \text{otherwise.}
    \end{cases}
  \]
\end{definition}

\begin{lemma}
  The lexicographic order~\cite[Sect.~2.4]{TermRewriting}, $>_{lex}$, on $(\omega+1)^\ast$
  is defined by
  \[
    (\omega+1)^n \ni (u_1, \ldots, u_n) >_{lex} (v_1, \ldots, v_m) \in (\omega+1)^m
  \]
  iff
  $n>m \vee (n=m \wedge \exists k<n. u_k = v_k \wedge \forall i< k. u_i > v_i)$,
  where $>$ is the natural order on $\omega+1$, that is $\omega > \ldots > 1 > 0$.
\end{lemma}

\begin{definition}
  The binary relation $>_\Measure$ on $\Measure$ is defined as follows.
  \par
  \begin{tabular}{rl}
    $M >_\Measure N$ iff 
  &
    there are $X, Y \in \Measure$ such that 
   $\emptyset \neq X \subseteq M$, \\&$N= (M \setminus X) \cup Y$, and
   $\forall y \in Y \exists x \in X. x >_{lex} y$.
  \end{tabular}
\end{definition}

\begin{lemma}
  The relation $>_\Measure$ is a strict and terminating order.
\end{lemma}
\begin{proof}
  We follow Baader and Nipkow~\cite{TermRewriting}.
  The natural order on $\omega+1$ is strict and terminating.
  Hence, so is $>_{lex}$~\cite[Lemma~2.4.3]{TermRewriting}.
  Therefore, the multiset order on $(\omega+1)^\ast$ is also strict and terminating~\cite[Lemma~2.5.4 and Theorem~2.5.5]{TermRewriting}.
  Due to the natural embedding of $\Measure$ into the set of finite mulisets on $(\omega+1)^\ast$,
  the relation $>_\Measure$ is dominated by the multiset order.  Hence $>_\Measure$ is terminating.
  Thanks to the same embedding, $>_\Measure$ is a strict order
\end{proof}

  We write $\geq_\Measure$ for the reflexive closure of $>_\Measure$.
  Similarly, $\leq_\Measure$ and $<_\Measure$ are meant.
%

\section{\texorpdfstring{$\emptyword$-Free}{Emptyword-Free} Formulas}
\label{sec:emptyword-free}

Formulas like $\EF^L \psi$ and $\AG^L \psi$ can speak about the current state if $\emptyword \in L$.
We intend to combine structures at their roots---in the proof to 
Thm.~\ref{thm:elimA} and~\ref{thm:wedgeEF elim}---, such that formulas should not realize this modification.
Nonetheless, formulas can be transformed accordingly.

\begin{definition}
  The property being \emph{$\emptyword$-free} is inductively defined on \PDLplain[$\cdot$]-formulas.
  \begin{enumerate} 
    \item Any literal is $\emptyword$-free.
    \item A conjunction and a disjunction is $\emptyword$-free if both con\-juncts or both dis\-juncts, respectively, are $\emptyword$-free.
    \item $\EF^L \varphi$ and $\AG^L \varphi$ are $\emptyword$-free iff $\emptyword \notin L$ and $\varphi$ is $\emptyword$-free.
  \end{enumerate} 
\end{definition}

\begin{definition}\label{def:ElimEW}
  The function $\ElimEW{\cdot}$ is defined on \PDLplain[$\cdot$]-formulas
   \begin{align*}
     \ElimEW{\ell} 
     &:= 
     \ell 
     &\text{ where $\ell$ literal or a constant} %
   \\
     \ElimEW{\varphi_0 \circ \varphi_0} 
     &:=
     \ElimEW{\varphi_0} \circ \ElimEW{\varphi_0} 
     &\text{ for }\circ\in\{\wedge,\vee\}
   \\
     \ElimEW{Q^{L} \varphi} 
     &:=
     \begin{cases}
       Q^{L} \ElimEW{\varphi}
       &\text{if }\emptyword \notin L
     \\
       \ElimEW{\varphi} \vee Q^{L\setminus\{\emptyword\}} \ElimEW{\varphi}
       &\text{otherwise if $Q=\EF$}
     \\
       \ElimEW{\varphi} \wedge Q^{L\setminus\{\emptyword\}} \ElimEW{\varphi}
       &\text{otherwise if $Q=\AG$}
     \end{cases}
     & \text{ for }Q\in\{\EF,\AG\}
   \end{align*}
\end{definition}

\begin{lemma}\label{lem:ElimEW}
  For every \PDLplain[$\cdot$]-formula $\varphi$ we have,
  \begin{enumerate}
    \item $\varphi$ and $\ElimEW{\varphi}$ are equivalent,
    \item $\ElimEW{\varphi}$ is $\emptyword$-free, and
    \item $\measure(\ElimEW{\varphi})  \leq_\Measure  \measure(\varphi)$.
  \end{enumerate}
\end{lemma}
\begin{proof} %
  Each item can be proven by induction on $\varphi$.  We detail on the last item
  for the second case of $\ElimEW{Q^{L} \varphi}$.
  As the IH yields $\measure(\ElimEW{\varphi}) \leq_\Measure \measure(\varphi)$, we have  %
  $%
    \measure(\ElimEW{\varphi}) 
    <_\Measure
    \measure(Q^{L \setminus\{\emptyword\}}\ElimEW{\varphi})
    \leq_\Measure
    \measure(Q^{L \setminus\{\emptyword\}}\varphi)
    \leq_\Measure
    \measure(Q^{L}\varphi)
  $. %
  Hence, the claim follows by 
  $
    \measure(\ElimEW{\varphi} \vee Q^{L\setminus\{\emptyword\}} \ElimEW{\varphi})
    \leq_\Measure
    \measure(Q^{L\setminus\{\emptyword\}} \ElimEW{\varphi})
  $.
\end{proof}

%
%
\section{Elimination of Outermost \texorpdfstring{$\AG$-Formulas}{AG-Formulas} and Conjunctions}
\label{sec:elim}

Although it is impossible to eliminate conjuncts and $\AG$-formulas in general,
the topmost ones can be removed (Thm.~\ref{thm:elimA} and~\ref{thm:wedgeEF elim}).
Hence, if $\vartheta$ is equivalent to $\EF^L \True$ for some language $L$ then $\vartheta$ can be rearranged
to a disjunction of $\EF$-formulas only.  However, these $\EF$-formulas might contain
conjunctions and $\AG$-formulas in turn.

\begin{definition}\label{def:DNF and completion}
  A formula $\vartheta$ is in \emph{disjunctive normal form (\DNF for short)} iff it has
  the shape
  \begin{align*}%
    \bigvee_{i \in I} \left(\bigwedge_{j \in J^\A_i} \alpha_{i,j} \wedge \bigwedge_{j \in J^\E_i} \varepsilon_{i,j} \right)
  \end{align*}
  where $I$, $J^\A_i$ and $J^\E_i$ are finite sets, $\alpha_{i,j}$ is an $\emptyword$-free
  $\AG$-formula, and $\varepsilon_{i,j}$ is an $\emptyword$-free $\EF$-formula (for all suitable indices).
  The \emph{completion of $\vartheta$} is
  \begin{align*}
    \Complete{\vartheta}
    &:=
    \vartheta \quad \vee \quad \bigvee_{\substack{\Psi' \subseteq \Psi \\ \models \bigwedge\Psi' \impl \vartheta}} \bigwedge\Psi'
  \end{align*}
  where $\Psi {:=} \{\varepsilon_{i,j} \mid i \in I, j \in J^\E_i\}$.
  A formula $\vartheta'$ is \emph{complete} iff it is $\Complete{\vartheta}$ for some $\vartheta$.
  The term ``\DNF'' and ``complete'' shall be applied up to associativity and commutativity of the 
  Boolean connectives\footnote{Note that this is well-defined when the measure $\measure$ is taken.}.
\end{definition}

\begin{lemma}\label{lem:DNF and completion}
  For any $\emptyword$-free formula $\vartheta$ we have
  \begin{enumerate}
    \item an equivalent formula $\vartheta'$ in \DNF 
       such that $\measure(\vartheta') \leq_\Measure \measure(\vartheta)$, and
    \item that $\Complete{\vartheta}$ is a \DNF, $\Complete{\vartheta}$ and $\vartheta$ are equivalent, and 
      $\measure(\Complete{\vartheta}) \leq_\Measure \measure(\vartheta)$.
  \end{enumerate}
\end{lemma}

\begin{proof}
  To get a \DNF, the distributive law is applied where $\AG$- and $\EF$-formulas are taken as atoms. 
  This application might rearranging (positive) Boolean connectives and might duplicate atoms.  However,
  the measure is defined in terms of unions for these cases.
  \par
  For the second item, the implication to $\vartheta$ follows from the definition of the additional disjuncts.
  The other direction is weakening.  As the additional terms are build only of top-level $\EF$-formulas in $\vartheta$,
  their measure is already subsumed in $\measure(\vartheta)$.  Note that $\mu$ is just the union in the
  case of the (positive) Boolean connectives.
\end{proof}

  For two structures $\mathcal M_1$ and $\mathcal M_2$ we define $\mathcal M_1 \oplus \mathcal M_2$
  as the \emph{disjoint sum} of both structures but with the root shared.
  The evaluation of the root is fixed as $\Prop = \emptyset$ for our purposes.
  The notation is extended to sequences of structures, say $(\mathcal M_i)_{i \in I}$, in the usual way, written as $\oplus_{i \in I} \mathcal M_i$.
  A formula $\psi$ is \emph{structurally monotone} iff for any model of $\psi$ any of its extension is also a model of $\psi$.
  An example is $\EF^L \True$ for any language $L$.

\begin{theorem}[Elimination of $\AG$-formulas]\label{thm:elimA}
  Let 
  \begin{align}
    \psi &:= \bigvee_{i \in I} \underbrace{\left(\alpha_i \wedge \bigwedge_{j \in J_i} \varepsilon_{i,j}\right)}_{=: \tau_i}
  \end{align}
  be complete where $I$, $J_i$ are finite, each $\alpha_i$ is a (possibly empty) conjunction of
  $\emptyword$-free
  $\AG$-formulas, and each $\varepsilon_{i,j}$ is a $\emptyword$-free $\EF$-formula.
  If $\psi$ is structurally monotone, then $\psi$ is equivalent to
  \begin{align}
    \psi' &:= \bigvee_{\substack{i \in I \\\models \alpha_i}} \underbrace{\left(\bigwedge_{j \in J_i} \varepsilon_{i,j}\right)}_{=: \tau_i'}
    \text.
  \end{align}
\end{theorem}

Note that $\measure(\psi') \leq_\Measure \measure(\psi)$.

\begin{proof}
  $\models \psi' \impl \psi$ is obvious.
  As the considered logic is closed under bisimulation, we consider tree-like structures in the following only.
  For the other direction, let $\mathcal M$ be a model of $\psi$.
  We have to show that $\mathcal M$ is also a model of $\psi'$.  
  If there is an $i \in I$ such that $\models \alpha_i$ %
  and $\mathcal M \models \tau_i$ then $\mathcal M \models \tau'_i$ and we are done.
  Otherwise, there is an $i_0 \in I$ such that $\not\models \alpha_{i_0}$ and 
  $\mathcal M \models \bigwedge_{j \in J_{i_0}}\varepsilon_{i,j}$, as $\mathcal M \models \psi$.
  For $i \in I$ define
  \begin{align}
    J^+_i &:= \{j \in J_i \mid \mathcal M \models \varepsilon_{i,j}\}
    \text{, and}
  \\
    J^-_i &:= J_i \setminus J^+_i
    \text.
  \end{align}  
  There are are two cases. Either
  \begin{align}\label{eq:elimA:resolution}
    \bigvee_{i \in I} \alpha_i\wedge\bigwedge_{j \in J^-_i} \varepsilon_{i,j}
  \end{align}
  is a tautology or not.
  If~\eqref{eq:elimA:resolution} is a tautology then so is
  \begin{align}\label{eq:elimA:resolvent}
     \bigwedge_{i \in I} \bigwedge_{j \in J^+_i} \varepsilon_{i,j}  \impl  \psi
  \end{align}
  as a simple case distinction on~\eqref{eq:elimA:resolution} shows.  
  Indeed, let $\widetilde{\mathcal M}$ be a model of the left side of~\eqref{eq:elimA:resolvent}.
  Then there is an $i \in I$ such that $\widetilde{\mathcal M} \models \alpha_i\wedge\bigwedge_{j \in J^-_i} \varepsilon_{i,j}$.
  Both together lead to $\widetilde{\mathcal M} \models \alpha_i\wedge \bigwedge_{j \in J_i} \varepsilon_{i,j}$
  and finally to $\widetilde{\mathcal M} \models \psi$.
  Hence, the left hand side of~\eqref{eq:elimA:resolvent} is a term in $\psi$ as the latter is complete.
  But, by definition, this term is modeled by $\mathcal M$.

  Otherwise~\eqref{eq:elimA:resolution} is not a tautology.  So there is a structure $\mathcal M'$ with
  \begin{align}
    \mathcal M' \not\models \alpha_i\wedge\bigwedge_{j \in J^-_i} \varepsilon_{i,j}
  \end{align} for all $i \in I$.
  We will exclude this situation.
  As $J^-_{i_0} = \emptyset$ we have $\mathcal M' \not\models \alpha_{i_0}$ in particular.
  Now let $\mathcal M'' := \mathcal M \oplus \mathcal M'$. %
  We claim that $\mathcal M'' \not\models \psi$ which is a contradiction
  to the assumption that $\psi$ is structurally monotone.  For the sake of contradiction, 
  suppose that there is an $i \in I$ such that $\mathcal M'' \models \tau_i$.  
  Among the $\EF$-formulas only those indexed by $J^+_i$ are already fulfilled in $\mathcal M$.  
  Hence, $\alpha_i \wedge \bigwedge_{j \in J^-_i} \varepsilon_{i,j}$ must be satisfied by $\mathcal M'$.    
  This is a contradiction to the choice of $\mathcal M'$.
  Note that we used implicitly that $\EF$-formulas are $\emptyword$-free.
\end{proof}

The theorem requires a syntactical presence of formulas called $\alpha_i$.
Note that minor changes make the proof also working if not all such parts a present.
On the other hand, inserting such an empty conjunction does not increase
the measure as atoms---such as $\True$---have the lowest measure anyway.

%
%
%
%
%
%
%
%
%
%
%
%
%
%
%
%
%
%
%
%
%
%
%
%
%
%
%
%
%
%
%

%
%

\begin{theorem}[Elimination of $\bigwedge\EF$-formulas]\label{thm:wedgeEF elim}
  Suppose 
  \begin{align}\label{eq:wedgeEF elim:1}
    \EF^L \True &= \delta \vee \bigwedge_{i \in I} \EF^{L_i} \psi_i
  \end{align} 
  where $\emptyword \notin L_i$ for all $i \in I$.
  If $I \neq \emptyset$ then there is an $i \in I$ such that
  \begin{align}\label{eq:wedgeEF elim:2}
    \EF^L \True &= \delta \vee \EF^{L_i} \psi
    \text.
  \end{align} 
\end{theorem}
Note that the measure of~\eqrefR{eq:wedgeEF elim:2} is bounded by that of~\eqrefR{eq:wedgeEF elim:1}, trivially.
\begin{proof}
  For any $i \in I$, \eqrefR{eq:wedgeEF elim:1} implies \eqrefR{eq:wedgeEF elim:2}.
  If there is an $i \in I$ with
  $\models \EF^{L_i} \psi_i \impl \EF^L\True$, this $i$ suffices for the other direction.
  To exclude the other case, assume that we have tree-like structures $\mathcal M_i$ for all $i \in I$ such that
  \begin{enumerate}
    \item $\mathcal M_i \models \EF^{L_i} \psi_i$ but
    \item $\mathcal M_i \not\models \EF^L\True$.
  \end{enumerate}
  Let $w_i \in L_i$ be the witness for the first item.
  Set $\mathcal M {:=} \oplus_{i \in I} \mathcal M_i$.
  The root of $\mathcal M$ might satisfy different formulas than the roof of $\mathcal M_i$,
  but this change is invisible to $\EF^{L_i} \psi_i$ since $|w_i| > 0$. 
  Hence, $\mathcal M \models \EF^{L_i} \psi_i$.
  For the sake for a contradiction, assume that $\mathcal M \models \EF^L \True$.
  This property depends only on a path in $\mathcal M$.  The path is inherited from some $\mathcal M_i$
  for $i \in I$.  Since $\EF^L \True$ does not depend on the evaluation of the root, $\mathcal M_i \models \EF^L\True$ which
  is a contradiction to the second property of $\mathcal M_i$.
  Therefore,  $\mathcal M \not\models \EF^L \True$.
  By construction we have $\mathcal M \models \bigwedge_{i \in I} \EF^{L_i} \psi_i$ but $\mathcal M \not\models \EF^L \True$.
  This property contradicts~\eqref{eq:wedgeEF elim:1}.
\end{proof}

\section{Extraction}
\label{sec:extraction}

In the proof of Theorem~\ref{thm:extraction} we apply previous elimination techiques to show
that the candidate formula is equivalent to $\bigvee_i \EF^{L_i} \psi_i$.
In the case that $L_i$ is not a singleton set, we cannot decompose $\psi_i$ any further.
Indeed, the proof relies on the property $w (w\backslash L) \subseteq L$ for any language $L$ and word $w$.  
However, this inequality is false when $w$ is replaced by a non-singleton language.
Nevertheless if the term $\EF^{L_i} \psi_i$ accepts a linear structure then
the term factorises the word on the structure. The left factor is $L_i$, surely,
and the right one can be read off as follows.

\begin{definition}\label{def:lang formula}
  Let $\varphi$ be a \PDLplain[$\cdot$]-formula.
  Its language is
  $%
    \Lang{\varphi} := \{w \in \Sigma^\ast \mid \pi_{w\;\$} \models \varphi \}
  $ %
  where $\pi_{a_1 \cdots a_n}$ is a path
  labeled with $a_1$ to $a_n$ for $a_1, \ldots, a_n\in\Sigma$. 
  The node reached after $n$ steps has no successor.
  In each node, no proposition holds.
\end{definition}

\begin{lemma}\label{lem:LangSound}
  $\Lang{\EF^{L\$}\True} = L$ for any $L \subseteq \Sigma^\ast$.
\end{lemma}
\begin{proof}%
  Let $w \in \Lang{\EF^{L\$}\True}$.
  By definition, we have $\pi_{w\$} \models \EF^{L\$}\True$. Hence there is a word $w' \in L$ such
  that $\pi_{w\$} \models \EF^{w'\$}\True$. As $w,w' \in \Sigma^\ast$ and $\$ \notin  \Sigma$, we have $w'=w$.
  For the converse, let $w \in L$, then $\pi_{w\$} \models \EF^{L\$}\True$. Hence $w \in \Lang{\EF^{L\$}\True}$.
\end{proof}

\begin{lemma}\label{lem:elimDollar}
  Let $L_0 \subseteq \Sigma_\$^\ast$, $L \subseteq \Sigma^\ast$, let $\delta$ be a formula, and let $\psi$ be a satisfiable formula.
  Suppose 
  \begin{align}\label{eq:elimDollar:1}
    \EF^{L\$} \True = \delta \vee \EF^{L_0} \psi
    \text.
  \end{align}
  Define $L_1 {:=} L_0 \cap \Sigma^\ast$ and 
  $L_2 {:=} \{w \in \Sigma^\ast\$ \mid w \text{ is a prefix of a word in } L_0\}$.
  Then 
  \begin{align}\label{eq:elimDollar:2}
    \delta \vee \EF^{L_0} \psi
    = 
    \delta \vee \EF^{L_1} \psi \vee \EF^{L_2}\True
    \text.
  \end{align}
  Additionally, the measure of~\eqrefR{eq:elimDollar:2}
  is weakly bounded by that of~\eqrefL{eq:elimDollar:2}.
\end{lemma}

\begin{proof}
\textbf{Case $\rightarrow$: }
Note that $\models \EF^{L_0} \psi \; \impl \;\EF^{L_1} \psi \vee \EF^{L_2}\True$ since for every word $w \in L_0 \setminus L_1$ there is a prefix of $w$ in $L_2$.
\textbf{Case $\leftarrow$: }
Since $L_1 \subseteq L_0$, $\models \EF^{L_1} \psi \; \impl \; \EF^{L_0} \psi$ holds. 
So, we assume a model $\mathcal M$ of $\EF^{L_2}\True$.
Hence there is a path $\pi$ in $\mathcal M$ labeled with a word $u\$ \in L_2$. 
Let $v \in \Sigma_\$^\ast$ such that $u\$v \in L_0$.
At the end of the path, we attach a path labeled with $v$ and on that one a model of $\psi$---note that $\psi$ is assumed to be satisfiable.
The new structure, say $\mathcal M'$, is a model of $\EF^{L_0} \psi$.
and also, by~\eqref{eq:elimDollar:1}, of $\EF^{L\$} \True$.
 
All (rooted) finite paths in $\mathcal M'$ which not yet occur in $\mathcal M$ passes the labels~$u\$$.
For the sake of contradiction, assume that $\mathcal M$ is not a model of $\EF^{L\$} \True$.
Hence $u\$$ is a prefix of a word in $L\$$.  So, $u \in L$ because $L \subseteq \Sigma^\ast$.
Contradiction.

And as for the measure,
\begin{alignat*}{2}
  \measure(\EF^{L_1}\psi) &= ||L_1|| :: \measure(\psi) &&\leq_\Measure ||L_0|| :: \measure(\psi) = \measure(\EF^{L_0} \psi) 
  \put(0,0){ and}
\\
  \measure(\EF^{L_2}\True) &= ||L_2|| \leq_\Measure ||L_0|| &&\leq_\Measure ||L_0|| :: \measure(\psi)
\end{alignat*}
hold, and imply $\measure(\EF^{L_1} \psi \vee \EF^{L_2}\True) \leq_\Measure \measure(\EF^{L_0} \psi)$.
\end{proof}

Two remarks, to previous lemma: (1) If $\emptyword \notin L_0$ then it is neither in $L_1$ nor in $L_2$---but $\$ \in L_2$ might be.
(2) If $L_0$ is a \DCFL then so are $L_1$ and $L_2$.

\begin{theorem}\label{thm:extraction}
  Let $P \subseteq \Sigma^\ast$, and let $\varphi$ a \PDLplain[\DCFL]-formula over $\Sigma_\$$.
  If 
  $%
    \varphi = \EF^{P \$} \True
  $ %
  then
  $\Lang{\varphi}$ is good.
\end{theorem}

\begin{proof}
  We apply to $\varphi$ several transformations in sequence.  Each transformation
  leads to a formula which is equivalent to $\varphi$ and 
  whose measure is weakly bounded by $\measure(\varphi)$ from above.
  The transformations are the following ones.
  \begin{itemize}
    \item Make the formula $\emptyword$-free by Definition~\ref{def:ElimEW} and Lemma~\ref{lem:ElimEW}.
    \item Transform it into a \DNF and complete the formula: Definition~\ref{def:DNF and completion} and Lemma~\ref{lem:DNF and completion}. 
    \item Eliminate the outermost $\AG$-quantifiers using Theorem~\ref{thm:elimA}.
    \item Apply Theorem~\ref{thm:wedgeEF elim} to each term of the \DNF gotten from the previous transformation.
      Note that the applied formula is still $\emptyword$-free.
    \item Apply Lemma~\ref{lem:elimDollar} and its remarks to the outermost $\EF$-formulas
  \end{itemize}
  Finally, we obtain a formula 
  \begin{align}\label{eq:extraction:varphi'}
    \varphi' {:=} \bigvee_{i \in I} \EF^{L_i} \psi_i \quad = \quad \varphi
  \end{align}
  such that $\measure(\varphi') \leq_\Measure \measure(\varphi)$ holds.
  In addition, $I$ is finite and for each $i \in I$ we have that 
  \begin{itemize}
    \item $L_i \subseteq \Sigma^+$, or $L_i \subseteq \Sigma^\ast\$$ and $\psi_i = \True$,
    \item $L_i \neq \emptyset$, and
    \item $L_i$ is a \DCFL.
  \end{itemize}

  If $L_i \subseteq \Sigma^+$ then set $L_i^\natural {:=} L_i$ and $R_i {:=} \Lang{\psi_i}$,
  and else, $L_i^\natural {:=} L_i / \$$ and $R_i {:=} \{\emptyword\}$.
  Note that $L_i^\natural$ is \DCFL in any case.

  \begin{claim}
    \label{claim:extraction:Lang varphi'}
    $\Lang{\varphi'} = \bigcup_{i\in I} (L_i^\natural \; R_i)$.
  \end{claim}
  \begin{proofofclaim}
  \textbf{$\subseteq$:} Let $w \in \Lang{\varphi'}$. By~\eqref{eq:extraction:varphi'},
  $\pi_{w\$} \models  \bigvee_{i \in I} \EF^{L_i} \psi_i$.
  By case distinction and by using that $\$$ is not part of $w$, we have $w \in \bigcup_{i\in I} (L_i^\natural \; R_i)$.
  \textbf{$\supseteq$:} Let $w \in \bigcup_{i\in I} (L_i^\natural \; R_i)$.
  We have to show that $\pi_{w\$} \models \bigvee_{i \in I} \EF^{L_i} \psi_i$.
  There are two cases.
  First, if $w \in L_i^\natural R_i$ for some $i \in I$ with $L_i \subseteq \Sigma^\ast \$$, then
  $\psi_i = \True$.  Hence $\pi_{w\$} \models \EF^{L_i} \psi_i$.
  Second, if $w \in L_i^\natural R_i$ for some $i \in I$ with $L_i \subseteq \Sigma^+$, then
  $w = uv$ with $u \in L_i$ and $v \in R_i$.  So, $\pi_{v\$} \models \psi_i$ and thus $\pi_{uv\$} \models \EF^{L_i} \psi_i$.
  \end{proofofclaim}

  Thus, $\Lang{\varphi'}$ is almost good.
  We have to exclude that there is an $i \in I$
  with $|L_i^\natural|=1$. 
  Let $I^+ {:=} \{ i \in I \mid |L_i| > 1\}$,
  $I^- {:=} \{ i \in I \mid |L_i| = 1\}$, and
  $I^-_a {:=} \{ i \in I^- \mid a \text{ is a prefix of the sole word in }L_i\}$ for $a \in \Sigma$.
  Let $\Sigma^- {:=} \{a \in \Sigma \mid I^-_a \neq \emptyset \}$.
  Note that $I = I^+ \cup I^-$ %
  and that $\{I^-_a\}_{a \in \Sigma^-}$ forms a partitioning of $I^-$.
  For $a \in \Sigma^-$, set 
  \begin{align*}
    \varphi_a {:=} & \bigvee_{i \in I^+} \EF^{a \backslash L_i}\psi_i \; \vee \; \bigvee_{i \in I^-_a} \EF^{a \backslash L_i} \psi_i 
    \text.
  \end{align*}
  As $a \backslash L_i = \emptyset$ for all $i \in I^-_b$ for $b \neq a$, the formula $\varphi_a$ is equivalent to $\EF^{a \backslash P \$} \True$.
  To apply the IH for $a \in \Sigma^-$, we have to ensure that $\measure(\varphi_a) <_\Measure \measure(\varphi')$.
  Indeed, $\measure(\EF^{a \backslash L_i}\psi_i) \leq_\Measure \measure(\EF^{L_i}\psi_i)$ for $i \in I^+$, and
  $\measure(\EF^{a \backslash L_i}\psi_i) <_\Measure \measure(\EF^{L_i}\psi_i)$ for $i \in I^-_a \neq \emptyset$.
  All in all, $\measure(\varphi_a) <_\Measure \measure(\varphi') \leq_\Measure \measure(\varphi)$ holds.
  We use the outcome of the IHs to replace the contributions of $I^-$ to $\Lang{\varphi'}$ by good languages.
  \begin{align*}
    P 
    =
    &
    \Lang{\varphi'}
    &
    \text{(by Lemma~\ref{lem:LangSound})}
  \\
    =
    & 
    \bigcup_{i\in I} (L_i^\natural \; R_i) 
    &
    \text{(by Claim~\ref{claim:extraction:Lang varphi'})}
  \\
    =
    &
    \bigcup_{i\in I^+} (L_i^\natural \; R_i)
    \cup 
    \bigcup_{a\in \Sigma^-}
    a \; \Big[
      \bigcup_{i \in I^+} a \backslash L_i^\natural R_i
      \; \cup \;
      \bigcup_{i \in I^-_a} a \backslash L_i^\natural R_i
    \Big]
  \\
    =
    &
    \bigcup_{i\in I^+} (L_i^\natural \; R_i)
    \cup 
    \bigcup_{a\in \Sigma^-} a \; [ a \backslash P]
  \\
    =
    &
    \bigcup_{i\in I^+} (L_i^\natural \; R_i)
    \cup 
    \bigcup_{a\in \Sigma^-} a \; \Lang{\varphi_a}
    & 
    \text{(by IH)}
  \end{align*}
  Also by IH, $\Lang{\varphi_a}$ is good.  So, $\Lang{\varphi'}$ is also good using the definition of $I^+$.
\end{proof}

\begin{corollary}\label{cor:XEF DCFL and Palindromes}
  Let $\varphi \in \text{\PDLplain[\DCFL]}$.
  If $\varphi = \EF^{\Palindromes\$} \True$ then the language $\Palindromes$ is good.
\end{corollary}

\begin{proof}
  By Theorem~\ref{thm:extraction} and Lemma~\ref{lem:LangSound}.
\end{proof}

\begin{corollary}\label{cor:XEF DCFL lneq XEF DCFL}
  \PDLplain[\DCFL] $\lneq$ \PDLplain[\CFL].
\end{corollary}
\begin{proof}
  By Corollaries~\ref{cor:XEF DCFL and Palindromes} 
  and~\ref{cor:Palindromes not good}.
\end{proof}

\section{Conclusion and Further Work}
\label{sec:conclusion}

We proved that \PDLplain[\DCFL] is distinct from \PDLplain[\CFL]
by means of model and language theory.
Similar results---such as \CTL vs.\ Fairness~\cite{EmersonHalpern86},
\PDLplain vs.\ \PDL~\cite{Berman81}, and unary \CTL vs.\ unary \CTLplus~\cite{EmersonHalpern85}---uses
two sequences of transition systems which are indistinguishable for the smaller logic.
Their proofs are pretty compact.  So, is it possible to reformulate our proof in a similar way?
The main difficulty should be the incorporation of Theorem~\ref{thm:dcfl neq palindromes} into transition systems.

The considered logic is exactly the $\EF$-/$\AG$-fragment of the Extended Computation Tree Logic~\cite{nonregctl_lpar_2010}, say \XCTL{$\mathfrak L$}.
This observation poses at least two  questions.
First, is it possible to extend the separation from the unary fragment to the binary $\EU$-/$\AR$-fragment?
Here, the main challenge is the interpretation of $\E(\psi_1 \U^L \psi_2)$ in the sense of Definition~\ref{def:lang formula}
as $\psi_1$ could prohibit linear models: take $\E((p \wedge (\EF^{\Sigma} \neg p)) \U^L \psi_2)$ for instance.
Secondly, one could go from one of these fragments to the whole logic to obtain a separation of \XCTL{\DCFL} and \XCTL{\CFL}.
In addition to the mentioned difficulties, one is faced with the alternating quantifiers $\EG$ and $\AF$.
To achieve such a goal, note that the Theorems~\ref{thm:elimA} and~\ref{thm:wedgeEF elim} also hold for arbitrary path quantifications as long as $\emptyword$-freedom is guaranteed. 
An iteration of these tools along a given $\omega$-word could unravel 
an $\omega$-sequence of disjunctions of $\E$-formulas.  
Such a sequence could be a subject for a pumping lemma similar to Lemma~\ref{lem:pumpinglemma}.
The $\omega$-word could follow the lines of Theorem~\ref{thm:dcfl neq palindromes}.

Finally, a separation of the full \PDL (i.e.\, with tests) and of the $\Delta$-variants of \PDL~\cite{Streett82,LLS07} 
could provide more insight into the difference between the non-determinism in \CFL{}s and the non-determinism 
used in the translation of formulas into automata.

\bibliographystyle{eptcs} 
\bibliography{main}

\end{document}